\documentclass[journal,final,twocolumn, letterpaper,10pt]{IEEEtran}

\usepackage[active]{srcltx} 
\usepackage{amsmath,amsthm,amsxtra,amssymb}
\usepackage{times,latexsym,array,verbatim,fancybox,multirow}
\usepackage{epsfig,graphicx,cite,url,setspace}
\usepackage[ruled,vlined,linesnumbered]{algorithm2e}

\def\BibTeX{{\rmfamily B\kern-.05em{\scshape i\kern-.025em b}\kern-.08em \TeX}}

\newtheorem{thm}{Theorem}[section]
\newtheorem{cor}[thm]{Corollary}

\theoremstyle{remark}
\newtheorem{rem}[thm]{Remark}

\newcommand{\Real}{\mathbb R}
\newcommand{\eps}{\varepsilon}

\newcommand{\mc}{\mathcal}

\newcommand{\D}{\mathcal{D}}

\newcommand{\g}{\gamma}
\renewcommand{\a}{\alpha}
\renewcommand{\l}{\lambda}

\title{\LARGE \bf A Unified Mechanism Design Framework for Networked Systems}

\author{Tansu Alpcan, Holger Boche, and Siddharth Naik  
\thanks{This work has been supported in part by Deutsche Telekom Laboratories.}
\thanks{Tansu Alpcan is with Technical University of Berlin, Deutsche Telekom Laboratories, in Berlin, Germany.
            {\tt\small alpcan@sec.t-labs.tu-berlin.de}}%
\thanks{Holger Boche and Siddharth Naik are with Technical University of Berlin, Heinrich Hertz Institute and Heinrich Hertz
Chair for Mobile Communication, respectively, in Berlin, Germany. 
        {\tt\small holger.boche@mk.tu-berlin.de} and  {\tt\small naik@hhi.fraunhofer.de}}%
}

\begin{document}
\maketitle
\thispagestyle{empty}


\begin{abstract}
Mechanisms such as auctions and pricing schemes are utilized to
design strategic (noncooperative) games for networked systems. 
Although the participating players are selfish, these mechanisms ensure 
that the game outcome is optimal with respect to a global criterion (e.g.
maximizing a social welfare function), preference-compatible, and
strategy-proof, i.e. players have no reason to deceive the
designer. The mechanism designer achieves these objectives by
introducing specific rules and incentives to the players; in this case
by adding resource prices to their utilities. In auction-based mechanisms, 
the mechanism designer explicitly allocates the resources based on bids of
the participants in addition to setting prices. Alternatively, pricing
mechanisms enforce global objectives only by charging the players for the 
resources they have utilized. In either setting, the player preferences represented
by utility functions may be coupled or decoupled, i.e. they depend on other player's
actions or only on player's own actions, respectively. 
The unified framework and its information structures are illustrated
through multiple example resource allocation problems from wireless
and wired networks.

\end{abstract}

\IEEEpeerreviewmaketitle

\section{Introduction} \label{sec:intro}

\textbf{Game theory} has been enjoying widespread adoption by the engineering community as a distributed 
optimization and control framework for networked systems, partly for taking into account preferences of individual
users, who share and compete for system resources. Resting upon a rich mathematical
foundation, game theoretical approaches, especially strategic (noncooperative) games, have been valuable for analysis and design of various resource allocation protocols in wireless and wired networks. Problems such as rate control, interference management, and power control (e.g. in wireless and optical networks)
have been investigated extensively by the research community using game theoretical methods
\cite{tac06-pavel,alpcan-twc,alpcan-ton,srikantbook}.
 
Game theory models nodes of networked systems as independent and autonomous decision makers with limited
global information, and studies incentives of individual players and effects of their preferences on the overall outcome. The Nash equilibrium (NE), where no player has an incentive to deviate from the NE while others adopt it, is known to be useful solution concept for such games. It is widely adopted for development of distributed and dynamic algorithms assuming some mild existence and uniqueness conditions \cite{tansuphd,basargame}.

Given the broad applicability of game theoretic frameworks, it is not surprising to observe an increasing interest in \textbf{mechanism design}, which studies rules and structure of games such that their outcome achieve certain objectives~\cite{maskin1,lazarSemret1998,johari1,caines1,alpcan-infocom10,hajek1}. 
This is especially relevant in development of distributed control schemes for networks where satisfying certain global 
properties such as \textbf{efficiency} are as important as the solution's compatibility with user incentives. 

A game designer can impose rules and incentives, e.g. in the form of prices, to players of such that the outcome of a
strategic game, for example, the unique Nash equilibrium solution is \textbf{preference-compatible} and at the same time maximizes a certain global objective function such as the sum of player utilities or 
quality-of-service (QoS) constraints. However, this interaction between the designer and  players of the game may create now a separate incentive for the players to misrepresent their utilities to the designer with the purpose of 
selfishly benefiting from it. Therefore, the mechanism designer has a third objective called \textbf{strategy-proofness} (or \textit{truth dominance}), in addition to the goals of efficiency and preference-compatibility.

This paper builds upon earlier work~\cite{cdc09lacra,gamenetsne}, which has presented a decision and control theoretic 
approach to game design taking into account only efficiency and preference-compatibility objectives while assuming that players are
 honest toward the game designer in terms of their preferences. Here, we present an optimization framework for mechanism design that satisfies all three objectives, adding strategy-proofness to the previous two.

The difficulty facing a mechanism designer trying to achieve all three objectives can be best appreciated with a specific example. Consider maximization of the sum of player utilities as the efficiency criterion of a specific problem. Assume that the designer
can impose a pricing scheme on users for their actions as an enforcement method. However, individual player utilities are not directly revealed to anyone. Assume in addition that the underlying
strategic game admits a unique NE solution. The task of the designer is then to find such a mechanism that it moves 
the NE of the game to a point, which maximizes the sum of these unknown player utilities
 (Figure~\ref{fig:neopt2}), while the players try to mislead the 
designer by misrepresenting their actual utility functions. In addition, the designer may not \textit{observe} players actions 
completely bringing additional restrictions to the information flow within the system.

Due to the difficulty of the above described task, there are naturally many impossibility results in the mechanism design 
literature \cite{holger-allerton,hurwicz1972,dasguptaHammondMaskin1979,zhou1991}. 
In contrast, this paper adopts a more constructive engineering approach and focuses on schemes that achieve all three objectives, albeit in some cases only approximately. The algorithms presented and analyzed here are examples of market clearance schemes, where all participants have an incentive to reveal their true preferences, and leading to solutions satisfactory to both designer and players from global and local points of view, respectively. Most of these mechanisms can be intuitively explained by the old adages of 
``\textit{actions speak louder than words}'' (designer deducing players' true preferences by observing their actions) and ``\textit{you get what you paid for}'' (designer charging players for their actions). 

However, we also note that the presented results are obtained only in very specific settings with various assumptions on player preferences (smooth and convex utility functions), on the underlying game (existence and uniqueness of the NE), and on a certain
degree of observability of player actions by the designer.  While these restrictions may decrease applicability of the results to certain
areas of economics, the presented optimization framework is of value in engineering settings, especially for the purpose of analyzing and developing distributed optimization and control schemes for networks.

The main contributions of this paper include:
\begin{itemize}
 \item Development of an unifying optimization framework for mechanism design, which encompasses both auction-based and
pricing mechanisms.
 \item Extension of earlier results on game design~\cite{cdc09lacra,gamenetsne} to mechanism design by taking into account
the strategy-proofness criterion.
 \item Application of the mechanism design framework to resource allocation problems in networks such
as rate control and interference management (power control).
\end{itemize}
The rest of the paper is organized as follows. The next section presents the underlying model and assumptions
of the unified framework developed. Section~\ref{sec:auction} studies auction-based mechanisms. Subsequently,
Section~\ref{sec:pricing} investigates pricing mechanisms. Section~\ref{sec:background}
provides an overview of relevant literature on mechanism design. The paper concludes with remarks of Section~\ref{sec:conclusion}.


\section{Unified Framework} \label{sec:gamedesign}

This section discusses the underlying model and assumptions of the unified framework for mechanism design.

\subsection{Model} \label{sec:model}

At the center of the game and mechanism design model is the \textit{designer} $\D$ who influences $N$  
 \textit{players}, denoted by the set $\mc A$, and participating in a \textbf{strategic (noncooperative) game}. 
These players are autonomous and independent decision makers, who share and compete for limited resources  under the given constraints of the environment. Concurrently,  the designer tries to ensure that the outcome of the game satisfies the desirable properties
of efficiency, preference-compatibility, and strategy-proofness. This setup is applicable to a variety of problems in networking 
(wireless spectrum and bandwidth management) and economics (auctions).

Let us define an $N$-player strategic game, $\mc G$, where 
each player $i \in \mc A$ has a respective  \textbf{decision variable} $x_i$ such that 
$$x=[x_1,\ldots,x_N] \in \mc X \subset \Real^N, $$ 
where $\mc X$ is the decision space of all players. As a starting point, this paper assumes scalar decision variables and a compact and convex decision space. The decision variables may represent, depending on the specific problem formulation, player flow rate, power level, investment, or bidding in an auction. Due to the inherent coupling between the players, the decisions of players directly affect each other's performance as well as the aggregate allocation of limited resources. 

The \textbf{preferences} of the players are captured by utility functions 
$$  U_i (x) : \mc X \rightarrow \Real, \;\;\; \forall i \in \mc A,$$ 
which are chosen to be continuous and differentiable for analytical tractability.
In many cases, the utility functions
have special properties such as concavity or monotonicity due to the underlying problem formulation, or these can be 
assumed to simplify the analysis.

The designer $\D$ devises a \textbf{mechanism} $M$, which can be represented by the mapping $M: \mc X \rightarrow \Real^N$, implemented by introducing incentives in the form of \textit{rules and prices} to players. The latter can be formulated by adding it as a cost term such that the player $i$ has the cost function
\begin{equation} \label{e:usercost}
  J_i(x)= c_i(x) - U_i (x) .
\end{equation}
Thus, the \textbf{player objective} is to  solve the following individual optimization problem in the strategic game
\begin{equation} \label{e:useropt}
 \min_{x_i} J_i(x) ,
\end{equation}
under the given constraints of the strategic game, and rules and prices imposed by the designer.
Specific properties and variants of these rules and prices will be discussed in the subsequent sections.

The \textbf{Nash equilibrium} (NE) is a widely-accepted and useful solution concept in strategic games, where no player has an incentive to deviate from it while others play according to their NE strategies. It plays an important role here since if it is unique, then the NE outcome
automatically satisfies the \textbf{preference-compatibility} criterion, which basically states that the mechanism outcome must coincide with 
the solution of the players' individual optimization problems (\ref{e:useropt}). 

The NE $x^*$ of the game $\mc G$ is formally defined as
$$ x_i^* := \arg \min_{x_i} J_i (x_i, x_{-i}^*) , \;\;\; \forall i \in \mc A,$$
where $x_{-i}^*=[x_1^*,\ldots,x_{i-1}^*,x_{i+1}^*,\ldots, x_N^*]$. The NE is at the same time the intersection point
of players' best responses obtained by solving  (\ref{e:useropt}) individually. If some special convexity and compactness conditions are imposed to the game $\mc G$, then it admits a unique NE solution, which simplifies mechanism and algorithm design significantly. For a detailed discussion on these conditions and properties of NE, we refer to \cite{tansuphd,basargame}.

Similar to player preferences, the \textbf{designer objective}, e.g. maximization of aggregate user utilities or social welfare, can be formulated using a smooth objective function $V$ for the designer:
$$ V(x,U_i(x),c_i(x)) : \mc X \rightarrow \Real,$$ 
where $c_i(x)$ and $U_i(x)$, $i=1,\ldots,N$ are user-specific pricing terms and player utilities,
respectively. Hence, the global optimization problem of the designer is simply $\max_x  V(x,U_i(x),c_i(x))$, which it solves
\textit{indirectly} by setting rules and prices.
In some cases, the objective function $V$ characterizes the desirability of an outcome $x$ from the designers perspective. In other cases when the designer objective is to satisfy certain minimum performance constraints
such as players achieving certain quality-of-service levels, the objective can be characterized by a region (a subset of the game domain $ 
\mc X$). Thus, the designer objective represents and corresponds to the \textbf{efficiency} criterion of the mechanism.

It is important to note that the designer can only influence the outcome of the game indirectly and cannot dictate actions of players (which would have immediately negated preference-compatibility). It has been shown in \cite{gamenetsne} that a function linear in $x_i$, such as $c_i(x)=\a_i x_i$, is sufficient for the designer to (indirectly) manipulate the unique NE outcome in the ideal full information case where the players are honest and open about their preferences. Figure~\ref{fig:neopt2} visualizes this process.
\begin{figure}[htp]
  \centering
  \includegraphics[width=0.6 \columnwidth]{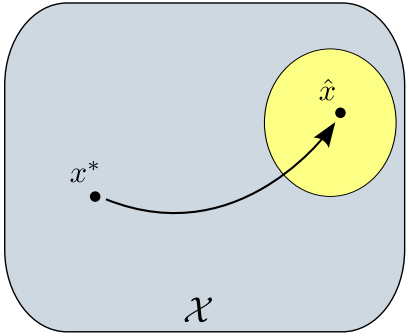}
  \caption{The manipulation of the unique Nash equilibrium, $x^*$ of the game by the mechanism designer $\D$ to a desirable region or point, $\hat x=\arg \max V$. }
\label{fig:neopt2}
\end{figure}

The third and an important criterion of mechanism design is \textbf{strategy-proofness}, which is also referred to as \textit{incentive-compatibility} or \textit{truth dominance}. If a mechanism does not 
possess this property, then the players have an incentive to misrepresent their utilities to the designer and ``\textit{cheat}'' in order to possibly obtain a larger share of the resources. Within the
context of the presented model, this criterion can be formally expressed as:
$$  J_i(x^*) < J_i(\tilde x) \Leftrightarrow c(x^*) - U_i (x^*)  < c(\tilde x) - \tilde U_i (\tilde x) \;\; \forall i \in \mc A,$$
where $\tilde U_i$ is the misrepresented utility, $x^*$ is the original NE solution, and $\tilde x$ is the distorted NE under $\tilde U_i$.
The interaction between the players of the underlying strategic game, $\mc A$, and the mechanism designer, $\D$ is depicted in Figure~\ref{fig:mechdesign1}.

\begin{figure}[htp]
  \centering
  \includegraphics[width=0.8 \columnwidth]{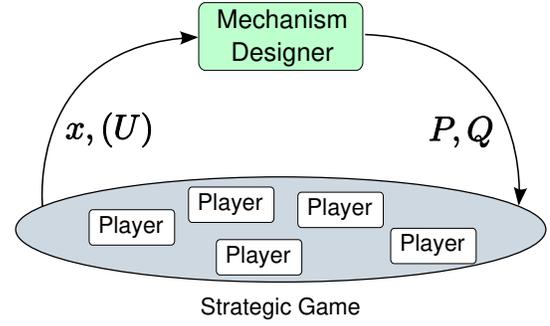}
  \caption{The interaction between the players of the underlying strategic game and the mechanism designer, who observes players actions
$x$ and utilities $U$, while imposing prices $P$ and in auctions an allocation rule $Q$.}
\label{fig:mechdesign1}
\end{figure}

\subsection{Assumptions}

Taking into account the breadth of the field mechanism design, it is useful to clarify the underlying assumptions of the model studied
in this paper. The \textbf{environment} where the players and designer interact is characterized by the following properties:
\begin{itemize}
 \item The available resources, which the players share and compete for, are limited.
 \item The environment imposes restrictions on available information to players and communication between them. Hence, it imposes a certain information structure to distributed mechanisms and sometimes makes it difficult to deploy centralized ones.
 \item The designer may not fully observe the player actions and has often limited information about their preferences.
\end{itemize}

The players share and compete for limited resources in the given environment under its information and communication constraints. Three basic types of resource sharing and coupling are often encountered in a variety of problems in networking:
\begin{enumerate}
 \item \textit{Additive resource sharing:} the players share a finite resource $C$ such that 
$$\sum_{i=1}^N x_i =C .$$
This type of coupling is encountered in bandwidth sharing and rate control in networks.
 \item \textit{Interference coupling} (linear interference): the resource allocated to player $i$, $\g_i$,  is inversely proportional to interference
generated others such that $$ \g_i(x) = \dfrac{h_i x_i}{\sum_{j \neq i} h_j x_j +\sigma}, $$
where $h_i \; \forall i$ and $\sigma$ denote some system parameters. Interference coupling occurs in wireless networks where $\g$ represents
signal-to-interference ratio.
 \item \textit{Multiplicative coupling:} the resource $y_i$ of player $i$ is affected multiplicatively by the decisions of others such that
$$ y_i= x_i \prod_{j \neq i} (1- x_j) .$$
 This type of coupling is seen in random multiple access schemes, e.g. slotted Aloha scheme in wireless networks \cite{tanenbaumbook}.
\end{enumerate}
It is possible to extend these definitions, for example, by making the finite resource $C$ time varying or changing the interference function. Couple of axiomatic frameworks for the second case exist in the literature  \cite{bocheSchubert_ITNet2008,yates1995}. The examples in this paper are of types 1 and 2.

The following assumptions are made on the designer and players:
\begin{itemize}
 \item The designer is honest, i.e. does not try to deceive the players.
 \item Each player acts independently and rationally according to its own self interests. 
 \item The players may try to deceive the designer by hiding or misrepresenting their individual preferences.
 \item Both players and designer follow the rules of the mechanism.
\end{itemize}

Within the scope of the model discussed in the previous subsection, specific formulations of the three criteria of mechanism design are summarized as:
\begin{table}[htp]
\begin{center}
\caption{Three Criteria of Mechanism Design  \label{tbl:mechdesign}}
\begin{tabular}[t]{|l|l|}
\hline 
 \textit{Criterion} & \textit{Formulation in the Model} \\
\hline \hline 
Efficiency   & Designer objective \\ 
Preference Compatibility & Player minimizing own cost \\ 
 & or existence of a unique NE \\  
Strategy-Proofness & No player gains from cheating \\
\hline
\end{tabular}
\end{center}
\end{table}

\section{Auction-Based Mechanisms} \label{sec:auction}

In auction-based mechanisms, the designer uses an allocation rule in addition to
pricing. Hence, the designer \textit{explicitly allocates} the players their share of
resources based on their bids. The players decide on their bids or actions by minimizing
their cost which is a combination of their own utilities and prices imposed by the designer. 
Specifically, the designer $\D$ imposes on a player $i \in \mc A$ a user-specific
\begin{itemize}
 \item resource allocation rule, $Q_i(x)$,
 \item resource pricing, $P_i(x)$,
\end{itemize}
where $x$ denotes the vector of player actions or bids. The specific properties
of these functions will be discussed later as part of individual mechanisms.

As presented in Section~\ref{sec:model},
each player $i$ aims to minimize its own cost $J_i(Q_i(x),P_i(x))$, as in (\ref{e:usercost}), while the designer
tries to achieve the objectives summarized in Table~\ref{tbl:mechdesign}. In some
cases, the designer may only observe the bids imperfectly as a function of the actual
bids, $y=f(x)$. However, in this paper, we assume that all bids are perfectly observable
and $y=x$ for simplicity. Figure~\ref{fig:auctionmech1} visually depicts the auction-based mechanisms described.
\begin{figure}[htp]
  \centering
  \includegraphics[width=\columnwidth]{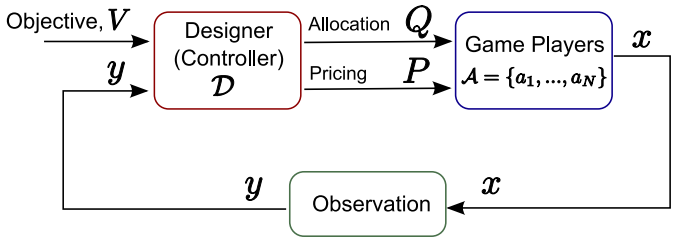}
  \caption{An auction-based mechanism, where the designer $\D$ imposes
a resource allocation rule as well as pricing on players $\mc A$ of the underlying
strategic game, whose bids $x$ may be observed imperfectly as $y$, with the
purpose of satisfying a global objective $V$. }
\label{fig:auctionmech1}
\end{figure}

\subsection{Auctions for Separable Utilities}

Consider, as a starting point, an additive resource sharing scenario where
the players bid for a fixed divisible resource $C$ and are allocated their share captured
by the vector $Q=[Q_1, \ldots, Q_N]$ such that at full utilization
$\sum_i Q_i =C$.

The $i^{th}$ player's individual cost function $J_i(x)$ in terms of  player bids $x$
is defined as
$$  J_i(x)= c_i(x) - U_i (Q_i(x)). $$
The pricing term has the general form of
\begin{equation} \label{e:genformofprice}
  c_i(x)= \int_0^{Q_i(x)} P_i(\xi) d\xi ,
\end{equation}
where $P_i$ denotes the unit price. In accordance with the earlier results~\cite{cdc09lacra,gamenetsne}
and due to the nature of the auction-based mechanism, it is sufficient for the purposes of the designer 
to choose a pricing function linear in $Q_i$,  i.e. $c_i(x)=P_i(x) Q_i(x)$. The player utility function $U_i$ is separable, i.e. it depends only on the individual allocation of the player. It is also assumed to be continuous, strictly concave, and twice differentiable in terms of its argument $Q_i$.
Thus, the cost function of player $i$ can be written as
\begin{equation} \label{e:playercost1}
 J_i(x)= P_i(x) \,Q_i(x) - U_i (Q_i(x)),
\end{equation}
which is strictly convex with respect to $Q_i$ under the assumptions made. 

From a player's perspective, who tries to minimize its cost in terms of the actual resources obtained, the condition
$$ \dfrac{\partial J_i}{\partial Q_i}=\dfrac{\partial c_i}{\partial Q_i} - \dfrac{\partial U_i}{\partial Q_i} =c_i^{\prime} - U_i^{\prime} $$
is necessary and sufficient for optimality. 
Thus  suppressing the dependence of user cost on bids $x$, 
in order for the auction-based mechanism to be \textbf{preference-compatible}, it has to satisfy
\begin{equation} \label{e:incentive1}
P_i(Q)=U_i^{\prime}(Q_i) \;\; \forall i \in \mc A. 
\end{equation}
Furthermore, if additional assumptions are made on $J_i(x)$, it can be shown that 
the game admits a unique NE, $Q^*$ (or $x^*$) \cite{basargame}.

Different from players, the designer $\D$ has two objectives: maximizing the sum of utilities of players and allocating 
all of the existing resource $C$, i.e. its full utilization. Hence, the designer $\D$ solves
the constrained optimization problem
\begin{equation} \label{e:designerobj1}
 \max_Q V(Q) \Leftrightarrow \max_Q \sum_i U_i (Q_i) \text{ such that } \sum_i Q_i=C ,
\end{equation}
in order to find a globally optimal allocation $Q$ that satisfies
this \textbf{efficiency criterion}.  
The associated Lagrangian function is then
$$ L(Q)=\sum_i U_i (Q_i) + \l \left( C- \sum_i Q_i \right)  ,$$
where $\l>0$ is a scalar Lagrange multiplier.
Under the convexity assumptions made, this leads to
\begin{equation} \label{e:global1}
\dfrac{\partial L}{\partial Q_i} \Rightarrow U_i^{\prime}(Q_i)= \l, \; \forall i \in \mc A,
\end{equation}
and the efficiency constraint
\begin{equation} \label{e:constraint1}
\dfrac{\partial L}{\partial \l} \Rightarrow \sum_i Q_i=C.
\end{equation}

\begin{rem}
It is important to note that sum of utility maximization as designer objective, 
i.e. $V=\sum_i U_i$ is  \textbf{only one} possible global objective among
many others such as ensuring a certain QoS to players 
(see \cite{gamenetsne,cdc09lacra} for a more detailed discussion).
\end{rem}

The interaction between the designer and players (see Figure~\ref{fig:mechdesign1}) is
through a \textit{bidding/allocation process} in auction-based mechanisms. Since
the players cannot obtain the resource $Q$ directly, they make 
a bid for their own total cost, which is denoted by the vector $x$. The pricing $P(x)$ 
and allocation $Q(x)$ rules of the auction-based mechanism should
satisfy the efficiency and preference-compatibility criteria discussed above.

A player's bid (or action), $x_i$, is an indicator of the player's willingness to pay and plays a crucial role in devising
a mechanism that is \textbf{strategy-proof}. Formally, a mechanism
is strategy-proof, if no player has an incentive to deviate from
its truthful bid
$$J_i(x_i^*+\delta) \geq J_i(x_i^*) \;\; \forall i \in \mc A, \, \delta,$$
where $\delta \in \Real$ is a scalar and $x^*$ is the outcome (NE) of the underlying
strategic game.

\subsection*{Example 1:}

In the specific resource sharing setting defined, an auction-based mechanism, $\mc M^a$,
can be defined based on the bid of player $i$,
\begin{equation} \label{e:bid1}
 x_i :=P_i(x) Q_i(x),
\end{equation}
the pricing function
\begin{equation} \label{e:price1}
 P_i:=\dfrac{ \sum_{j \neq i} x_j + \omega}{C},
\end{equation}
for a scalar $ \omega>0$ sufficiently large such that $\sum_i Q_i \leq C$,
and the resource allocation rule
\begin{equation} \label{e:Q1}
 Q_i:=  \dfrac{x_i}{\sum_{j \neq i} x_j + \omega}\, C.
\end{equation}
It is also possible to interpret the scaler $\omega$ as a \textit{reserve bid} \cite{huangBerryHonig2006b}.
The next theorem establishes that this mechanism is preference-compatible, strategy-proof, and asymptotically
efficient.

\begin{thm} \label{thm:auction1}
The auction-based mechanism $\mc M^a$ defined by (\ref{e:bid1}), (\ref{e:price1}), and (\ref{e:Q1})
allocates the fixed divisible resource $C$ to a set of selfish rational players $\mc A$ with respective
cost functions (\ref{e:playercost1}) in such a way that the mechanism is preference-compatible, 
strategy-proof, and asymptotically efficient, if
$$ U_i(x_i^*+\delta) - U_i(x_i^*) \leq \delta, \;\; \forall i, \, \forall  \delta \in \Real,$$
where $x^*$ denotes the truthful bid of player $i$ at the NE outcome.
In other words, the outcome of the mechanism ensures that
\begin{itemize}
 \item optimal allocation obtained, $Q^*$, satisfies  \\ 
 $ Q_i^*= \arg \min_{Q_i} J_i (Q) \Rightarrow P_i(Q^*)=U_i^{\prime}(Q_i^*) \;\; \forall i$,
 \item no player has an incentive to deviate from
its truthful bid, $J_i(x_i^*+\delta) \geq J_i(x_i^*), \;\; \forall i, \, \delta$
 \item $Q^*$ solves the constrained optimization problem in (\ref{e:designerobj1}) asymptotically, 
 i.e. as $\lim N \rightarrow \infty$.
\end{itemize}
\end{thm}
\begin{proof}

The mechanism $\mc M^a$ is defined by the bidding process (\ref{e:bid1}), unit prices (\ref{e:price1}), and 
allocation rule (\ref{e:Q1}) for each player $i \in \mc A$. Substituting these into the player cost function (\ref{e:playercost1}) 
results in
$$  J_i (x) = x_i - U_i \left( \dfrac{x_i}{\sum_{j \neq i} x_j + \omega}\, C \right) .$$
Due to the convexity of $J_i$ in $x_i$, the first order necessary condition for optimality is also sufficient:
$$ \dfrac{\partial J_i (x)}{\partial x_i }= 1 -  
\dfrac{\partial U_i (Q_i)}{\partial Q_i} \left( \dfrac{C}{\sum_{j \neq i} x_j + \omega}\right)=0.$$
From definition of $P_i$ in  (\ref{e:price1}) follows $P_i(Q)=U_i^{\prime}(Q_i)$ for each
player $i$. Hence, the mechanism $\mc M^a$ is \textit{preference-compatible}. Furthermore,
it is straightforward to show that this game admits a unique NE, $x^*$ \cite{basargame}.

Assume a player $i$ deviates from its truthful bid $x_i$ by an amount $\delta \in \Real$ such
that $\tilde x_i = x_i + \delta$. Then, the player cost under $\mc M^a$ becomes
$$\tilde J_i= \tilde x_i - U_i \left( \dfrac{\tilde x_i}{\sum_{j \neq i} x_j + \omega}\, C \right) .$$
In order $\mc M^a$ to be \textit{strategy-proof}, 
$$\tilde J_i - J_i = \delta -  \left( U_i (Q_i(x_i+\delta)) - U_i (Q_i(x_i)) \right) >0 ,$$
which immediately holds under the assumption in the theorem.

Although it is preference-compatible and strategy-proof, the mechanism $\mc M^a$ 
is not fully efficient as it does not exactly solve the designer optimization
problem (\ref{e:designerobj1}). To see this, let us solve (\ref{e:global1})
and (\ref{e:constraint1}) using $x_i = P_i Q_i$ to obtain
$$ P_i = \dfrac{\sum_i x_i }{C} \text{ and } Q_i = \dfrac{x_i}{\sum_i x_i }C .$$
These optimal solutions (with respect to designer objective) are  only approximated by the pricing (\ref{e:price1}) 
and allocation  (\ref{e:Q1}) rules. Hence, 
$$  P_i= \dfrac{\sum_{j \neq i} \a_j +\omega}{C} \neq \dfrac{\sum_i \a_i}{C} $$
and
$$ \sum_i Q_i =\sum_i \dfrac{ \a_i}{\sum_{j \neq i} \a_j +\varepsilon} C \approx C. $$
The choice of suboptimal (in the sense of efficiency) rules is due to the fact that $\mc M^a$ has to achieve 
strategy-proofness at the same time as efficiency and preference-compatibility.
However, as the number of players increases, $N\rightarrow \infty$, and by choosing $\omega$ accordingly small, the approximation
becomes more precise. Thus, the mechanism $\mc M^a$ is \textit{asymptotically efficient}.
\end{proof}

\subsection*{Example 2:}

As a special case of the auction-based mechanism $\mc M^a$, consider a setup where,
the player utility functions are logarithmic and respectively weighted by a positive scalar parameter $\a$ such that
$$U_i=\a_i \log Q_i \;\; \forall i \in \mc A. $$ 
Then, the following result holds as a special case of Theorem~\ref{thm:auction1}.

\begin{cor} \label{thm:auction2}
The auction-based mechanism $\mc M^a$ defined by (\ref{e:bid1}), (\ref{e:price1}), and (\ref{e:Q1})
allocates the fixed divisible resource $C$ to a set of selfish rational players $\mc A$ with respective
cost functions (\ref{e:playercost1}) and utilities $U_i=\a_i \log Q_i \;\; \forall i \in \mc A$ in such a way that the 
mechanism is preference-compatible, strategy-proof, and asymptotically efficient.
\end{cor}
\begin{proof}
The proofs of preference-compatibility and asymptotic efficiency follow directly from the ones of
Theorem~\ref{thm:auction1}. Furthermore, the mechanism is strategy-proof under logarithmic
player utilities since they satisfy the sufficient condition in Theorem~\ref{thm:auction1}.
The condition in this case is
$$ \a_i \log (Q_i(x_i+\delta) ) - \a_i \log (Q_i(x_i)) \leq \delta, $$
leading to
$$ \log \left(1+\frac{\delta}{x_i}\right) \leq \dfrac{\delta}{\a_i}. $$
The player's truthful bid is $x_i=\a_i$ from its cost function (\ref{e:playercost1}). Thus, 
we obtain
$$ \exp\left(  \dfrac{ \delta}{\a_i }\right)  > 1+ \dfrac{\delta}{\a_i},$$
which holds by definition, and completes the proof.
\end{proof}

\subsection{Auctions for Non-separable Utilities} \label{sec:auctionnonsep}

In many problem formulations, the player utilities are non-separable, i.e. they depend also on other
player's actions. This is the case, for example, in interference coupled systems such as a cellular wireless
system with a base station (acting as the designer) and mobile devices or users as players who bid to achieve
a certain QoS level. Let $x_i$ denote the bid of a mobile device and the $q_i(x)$ the transmission power
assigned to it by the base station. Then, the 
signal-to-interference ratio (SIR) of the received signal by the mobile is
\begin{equation} \label{e:sir}
  \g_i =\dfrac{q_i(x)}{\sum_{j \neq i} q_j(x) + \sigma} ,
\end{equation}
where $\sigma>0$ is an independent noise term. Notice that this is essentially a centralized
scheme similar to the ones currently deployed. A decentralized version will be discussed
in Section~\ref{sec:pricing}.

This interference management and power control formulation has been discussed extensively 
in the literature, e.g. \cite{bocheSchubert_ITNet2008,yates1995,alpcan-winet2}. However,
such mechanisms do not necessarily need to be limited to wireless
networks and apply to any system with linear interference coupling~\cite{disc09}
under the assumption that the player utilities are $U_i(\g_i)$ continuous,
strictly concave, and twice differentiable in their arguments $\g_i$ (\ref{e:sir}).

\subsection*{Example 3}

Consider an auction-based mechanism for an interference-coupled system
where players have non-separable and logarithmic utilities and a linear pricing scheme, which
make the problem more tractable. Then, each player $i$ minimizes its respective cost 
\begin{equation} \label{e:playercost4}
 J_i(x)= P_i(x) q_i(x) - \a_i \log (\g_i(q(x))) ,
\end{equation}
which is strictly convex in player power level $q_i$. Consequently, the general condition
for player preference-compatibility is $ P_i = \a_i / q_i, \; \forall i \in \mc A$, 
as in Examples 1 and 2. 

The global objective of the designer is to maximize sum of utilities of players while trying
to limit the total interference effect to an upper-bound $C$. This approximate formulation 
is motivated by, for example, limiting the aggregate inter-cell interference created by the mobile devices 
in a wireless network, where base stations have no means of communication among themselves.
Hence, the designer $\D$ solves the constrained convex optimization problem
$$ \max_q V(q) \Leftrightarrow \max_q\sum_i \a_i \log (\g_i(q)) \text{ such that } \sum_i q_i \leq C .$$

The resulting necessary and sufficient conditions for optimality are
$$ \dfrac{\a_i}{q_i} - \sum_{j \neq i} \dfrac{\a_j}{\bar C - q_j}= \l \text{ and } \sum_i q_i = C,$$
where $\bar C=C+\sigma$.

In the specific resource sharing setting defined, an auction-based mechanism, $\mc M^b$,
is defined based on the bid of player $i$,
\begin{equation} \label{e:bid2}
 x_i :=P_i(x) Q_i(x).
\end{equation}
and the allocation rule 
\begin{equation} \label{e:Q2}
 Q_i:= \dfrac{x_i}{P_i(x)}=q_i(x) ,
\end{equation}
which assigns users power levels based on their bids and computed prices.

Under the preference-compatibility condition, the bids
have to match the utility parameter, $x_i=\a_i$. Then,
the optimality conditions for the global problem become
\begin{equation} \label{e:globalaucnonsep1}
\dfrac{x_i}{q_i} - \sum_{j \neq i} \dfrac{x_j}{\bar C - q_j}= \l \text{ and } \sum_i q_i = C.
\end{equation}
which are solved to obtain $(q^*, \l^*)$. 
Accordingly, the pricing function is
\begin{equation} \label{e:price2}
 P_i(x):=\l^* +   \sum_{j \neq i} \dfrac{x_j}{\bar C - q_j^*}.
\end{equation}
As a result of this design, the auction-based mechanism $\mc M^b$ is clearly
efficient and preference-compatible.

We next show that  mechanism $\mc M^b$ is asymptotically strategy-proof. Assume that 
a player $i$ deviates from its truthful bid $x_i$ by an amount $\delta \in \Real$ such
that $\tilde x_i = x_i + \delta$. The strategy-proofness is then equivalent to
$$ \tilde J - J = \delta - \a_i \log \left( \dfrac{\g_i (x_i +\delta)}{\g_i (x_i)}\right)  >0.$$
As in the previous example, this leads to
$$  \dfrac{\g_i (x_i +\delta)}{\g_i (x_i)} < \exp (\dfrac{\delta}{\a_i}), $$
or
$$  \dfrac{\tilde q_i }{q_i} \dfrac{\bar C - q_i}{\bar C -\tilde q_i} \dfrac{\l}{\tilde \l}
 < \exp (\dfrac{\delta}{\a_i}), $$
where $\tilde \l$ is the solution of (\ref{e:globalaucnonsep1}) under $\tilde x_i$.
Note that, as the number of players goes to infinity,\footnote{We remind here the
underlying assumption that each player acts individually and there is no coordination among players.
This assumption is applicable to many networked systems with information flow constraints. }
we have
$$ \lim_{N \rightarrow \infty} \dfrac{\bar C - q_i}{\bar C -\tilde q_i} \dfrac{\l}{\tilde \l} =1.$$
Thus, it asymptotically holds that
$$ 1 + \dfrac{\delta}{\a_i}< \exp (\dfrac{\delta}{\a_i}),$$
which establishes the result summarized in the following theorem.

\begin{thm} \label{thm:auction3}
Consider a set selfish rational players $\mc A$ with respective cost functions (\ref{e:playercost4}) and non-separable utilities $U_i=\a_i \log(\gamma(q(x))) \;\; \forall i \in \mc A$ in an interference-coupled
system (\ref{e:sir}). 
The auction-based mechanism $\mc M^b$ defined by (\ref{e:bid2}), (\ref{e:Q2}), and  (\ref{e:price2})
maximizes the sum of utilities of players while limiting the total interference effect to 
an upper-bound $C$ in such a way that the 
mechanism is preference-compatible, efficient, and asymptotically strategy-proof.
\end{thm}


%
%
%

\section{Pricing Mechanisms} \label{sec:pricing}

Pricing mechanisms differ from auction-based ones by the property that the designer does
not allocate the resources explicitly, i.e. there is no allocation rule $Q$. 
\textit{The players obtain resources directly as a result of their actions} 
but are charged for them by the designer observing these actions (Figure~\ref{fig:pricingmech1}).
Hence, the designer has relatively less leverage in this case compared to auctions.

Pricing mechanisms are applicable to many networked systems where an explicit allocation of resources brings a prohibitively expensive
overhead or simply not feasible, e.g. due to participating players being selfish or located in a distributed manner. 
Example problems include rate control in wired networks, interference management in wireless networks,
and power control in optical networks \cite{tac06-pavel,alpcan-twc,alpcan-ton,srikantbook}.

\begin{figure}[htp]
  \centering
  \includegraphics[width=\columnwidth]{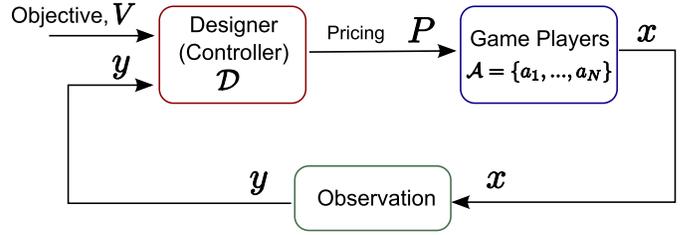}
  \caption{The block diagram of a generic pricing mechanism. The designer sets the prices $P$ to achieve
a global objective $V$ based on the observations $y$ of player actions $x$. The players choose their
actions $x$ independently according to their utilities $U$ (preferences) and prices $P$. The overall
mechanism aims to ensure efficiency, preference-compatibility, and strategy-proofness.}
\label{fig:pricingmech1}
\end{figure}

\subsection{Pricing Mechanisms for Separable Utilities}

We study an additive resource sharing scenario, where
the players compete for a fixed divisible resource $C$ as in Section~\ref{sec:auction}. 
The players' individual cost functions, which they minimize, have the general form
\begin{equation} \label{e:playercost2}
 J_i(x)= P_i(x) x_i  - U_i (x_i).
\end{equation}
Here, $x_i$  denotes the player's action of obtaining that specific amount of the resource directly,
in contrast to bidding for it and receiving an allocation from the designer.
It is sufficient for the purposes of the designer to choose a pricing function linear in $x_i$. A more general form
of pricing is provided in (\ref{e:genformofprice}). The player utility function $U_i$ is assumed
to be continuous, strictly concave, and twice differentiable. At the same time it only takes the player's own
action as its argument, i.e. the \textit{player utilities are separable} in this formulation.

In order for a pricing mechanism to be \textit{preference-compatible}, it has to satisfy
$$ P_i(x^*)=U_i^{\prime}(x_i^*), \;\; \forall i \in \mc A, $$
which directly follows from (\ref{e:playercost2}). The point $x^*$ is, by definition,
the Nash equilibrium solution of of the strategic game, where no player has an incentive to deviate from
it. Under the assumptions made for player utilities, the game admits a unique Nash equilibrium
solution \cite{basargame}. 
It is important to note that, if there was no pricing term in (\ref{e:playercost2}), each player would
try to get a large proportion of the resource resulting in a suboptimal result for everyone;
a situation sometime termed as  \textit{tragedy of commons}. The designer can prevent this by
a carefully selected pricing scheme \cite{cdc09lacra,gamenetsne}.

The global objective of the designer can be maximization of the sum of player utilities while
ensuring full resource usage, i.e. $\sum_i x_i=C$. Hence, the designer $\D$ solves
the counterpart of the constrained optimization problem in (\ref{e:designerobj1}) along
with (\ref{e:global1}) and (\ref{e:constraint1}). 

When the two criteria of preference-compatibility and efficiency (designer objectives) are combined,
the pricing function $P_i$ of a player $i$ has to satisfy
$$  P_i(x^*)=U_i^{\prime}(x_i^*)= \l, \; \forall i \in \mc A, $$
where $\l>0$ is the unique Lagrange multiplier. From the criterion of full resource usage,
it follows that
\begin{equation} \label{e:m2}
  \sum_i  x_i^* = \sum_i \left( U_i^{\prime}\right)^{-1}(\l) =C.
\end{equation}
Define  $\l^*$ as the optimal solution to (\ref{e:m2}) given player utilities $U_i$ and capacity $C$.
Then, the optimal pricing function is: $ P_i=\l^* \; \forall i$. 

If the designer wants to compute the unit prices $P$ directly by solving (\ref{e:m2}), it needs to ask the individual players
for their utilities. However, the players have an incentive to misrepresent their utilities to gain a larger
share of resources, if they are asked directly by the designer. 
Such a direct mechanism has two significant disadvantages. First, the designer has to have additional schemes in place
to detect potential player misbehavior (for which players have an incentive). Second, 
it brings another layer of communication overhead to the system. The disadvantages
of such direct mechanism will be illustrated more concretely in the scope of an example in the next subsection.

Alternatively, one can design an \textbf{iterative pricing mechanism} that is based
on observation of player actions $x$ instead of asking for their word (utilities). Then,
the designer deploys this iterative mechanism 
to compute the optimal prices $P_i=\l^*$ as a solution to (\ref{e:m2}).

For example, consider the following iterative pricing mechanism
\begin{equation} \label{e:iterative1a}
 \l (n+1) = \l (n) + \kappa \left( \sum_i x_i - C \right) ,
\end{equation}
where $\kappa >0$ is a small step size, $\l>0$, and
\begin{equation} \label{e:iterative1b}
 x_i (n+1) = \phi x_i (n) + (1-\phi)  \left( \dfrac{\partial U_i}{\partial x_i} \right)^{-1}(\l), \; \;\; \forall i \in \mc A,
\end{equation}
where $0<\phi<1$. Here, $n\geq 1$ denotes the time (update) step. Note that, the players
adopt a relaxed or gradient update scheme instead of best response taking into account variability of
the system. The gradient update also helps with convergence. 


\subsection*{Example 4:}

As a special case, let the utility function of players be logarithmic and weighted by parameter $\a$ such that
$$U_i(x_i)=\a_i \log x_i$$ 
for player $i$. Such utility functions have been utilized in the literature, for example, 
to model user demand in rate or congestion control on networks. The solution aligning the player and designer
objectives, in other words  the efficient Nash equilibrium, has the following properties:
$$ P_i^*=\dfrac{\a_i}{x_i}=\l ; \quad x_i=\dfrac{\a_i}{\l} \;\; \forall i \in \mc A$$
$$\Rightarrow \sum_i x_i=\dfrac{\sum_i \a_i}{\l}=C ; \quad \l = \dfrac{\sum_i \a_i}{C}. $$
Hence, the resulting optimal pricing mechanism for all players is 
\begin{equation} \label{e:exp2}
P= \dfrac{\sum_i \a_i}{C} . 
\end{equation}

Although this solution is preference-compatible from the players' perspective and
solves the global optimization problem of the designer, it is not strategy proof if the
designer explicitly asks the players for their utility parameter $\a$. To see this,
assume that player $i$ has a true utility parameter $\a_i$ but misrepresents
it to the designer as $\tilde \a_i=\a_i + \delta$ for some $\delta \in \Real$. 
Then, the new price is  $\tilde P= (\sum_i \a_i + \delta)/C$ and player $i$
real cost becomes
$$\tilde J_i(\tilde x_i, x_{-i})= \tilde P \tilde x_i - \a_i \log (\tilde x_i)$$
instead of 
$$ J_i(x) =  P  x_i - \a_i \log (x_i).$$
Substituting $\tilde P$ and computing $\tilde x_i$ yields
$$ \tilde J_i(\tilde x_i, x_{-i})= \a_i - \a_i \log \left( \dfrac{\a_i}{\sum_i \a_i + \delta} C \right),$$
and similarly we have
$$ J_i (x)= \a_i - \a_i \log \left( \dfrac{\a_i}{\sum_i \a_i } C \right).$$
Clearly, the player $i$ can decrease its cost ($\tilde J_i < J_i$) by choosing a $\delta <0$
despite being charged the same total price. Thus, the mechanism is not strategy-proof.

This issue is remedied by adopting the proposed iterative pricing mechanism:
\begin{equation} \label{e:iterative2a}
 \l (n+1) = \l (n) + \kappa \left( \sum_i x_i - C \right) ,
\end{equation}
\begin{equation} \label{e:iterative2b}
 x_i (n+1) = \phi x_i (n) + (1-\phi)  \dfrac{\a_i}{\l} \;\; \forall i \in \mc A.
\end{equation}
The unique (Nash) equilibrium solution of this iterative algorithm, $(x^*, \l^*)$
solves the designer problem (\ref{e:designerobj1}). Furthermore, since the players
adopt here a relaxed (gradient) best response at each step and there is no
explicit communication between the players and the designer, the scheme
is strategy-proof. To see this, assume otherwise and let player $i$ ``misrepresent''
its actions $\tilde x_i = x_i + \delta$ for some $\delta \in \Real$. Then, the player's
instantaneous cost is $\tilde J_i > J_i$ at each step of the iteration. Hence, the
players have no incentive to ``cheat''.

The communication requirements of the algorithm (\ref{e:iterative2a})-(\ref{e:iterative2b}) 
are minimal and suitable for a distributed implementation in a networking environment. 
The designer only needs to observe the total amount $y=\sum_i x_i$ and communicate
the common price $P$ back to the players (see Figure~\ref{fig:pricingmech1}
for a visualization).

Now, a basic stability analysis is provided for the following continuous-time approximation of
the iterative pricing mechanism 
$$ \dot \l = \dfrac{d \l}{dt}= \kappa \left( \sum_i x_i - C \right) , $$
$$ \dot x_i =- \dfrac{\partial J_i}{\partial x_i }= 
\bar \kappa_i \left(\dfrac{\a_i}{x_i} - \l \right), $$
where $t$ denotes time and $\bar \kappa_i>0$ is a user-specific step size. As in the discrete-time
version, the players adopt 
here a gradient best response algorithm. Define the Lyapunov function
$$ V_L:= \frac{1}{2} \left(\sum_i x_i -C\right)^2 +\frac{1}{2} \sum_i  \left(\dfrac{\a_i}{x_i} - \l \right)^2 ,$$
which is nonnegative and satisfies $V_L (x^*,\l^*)=0$. It is straightforward to show
through algebraic manipulations that $\dot V_L (x,\l)<0$ for all $(x,\l) \neq (x^*,\l^*)$. 
Hence, the continuous-time algorithm is globally asymptotically stable \cite{khalilbook}.
This result is a strong indicator of convergence \cite{bertsekas3} of the discrete-time iterative pricing mechanism (\ref{e:iterative2a})-(\ref{e:iterative2b}).

\subsection{Pricing Mechanisms for Non-separable Utilities} \label{sec:pricenonsep}

In some problem formulations, such as interference coupled systems consisting of a base station (acting as the designer) and mobile devices as players, the players' actions are beyond the control of the
base station. Let, specifically, $x_i = h_i p_i$ denote the received power level as a product
of uplink transmission power $p_i$ and channel loss $0<h_i<1$ of player $i$. If linear interference is assumed, then the signal-to-interference ratio (SIR) of the received signal is
\begin{equation} \label{e:sir2}
  \g_i =\dfrac{x_i}{\sum_{j \neq i} x_j + \sigma} ,
\end{equation}
as in (\ref{e:sir}). 

In pricing mechanism, similar to the auction in Section~\ref{sec:auctionnonsep}, each player $i$ minimizes its respective cost 
\begin{equation} \label{e:playercost3}
 J_i(x)= P_i(x) x_i - \a_i \log (\g_i(x)) ,
\end{equation}
which is strictly convex in $x_i$. Consequently, the general condition
for player preference-compatibility is  $ P_i = \a_i / x_i, \; \forall i \in \mc A$.

The global objective of the designer aims to maximize sum of utilities of players while trying
to limit the total interference effect to $C$, motivated by e.g. limiting the
aggregate interference created by the mobile devices in a wireless network.
Hence, the designer $\D$ solves
$$ \max_x V(x) \Leftrightarrow \max_x \sum_i \a_i \log (\g_i(x)) \text{ such that } \sum_i x_i \leq C .$$
This problem differs from the one in Section~\ref{sec:auctionnonsep} as it is non-convex.
However, it can be convexified using the nonlinear
transform $x_i=e^{s_i}$, and then admits a unique solution \cite{bocheSchubert_ITNet2008}.

The resulting necessary and sufficient conditions for optimality are
$$ \dfrac{\a_i}{x_i} - \sum_{j \neq i} \dfrac{\a_j}{I_j}= \l \text{ and } \sum_i x_i = C,$$
where $I_i:=\sum_{j \neq i} x_j + \sigma$ is the interference affecting player $i$.
Hence, aligning the player and designer optimization problems leads to
$$ P_i =  \l + \sum_{j \neq i} \dfrac{\a_j}{I_j}.$$
Using the definition of $\g$, this can be rewritten as
$$ P_i =  \l + \sum_{j \neq i} P_j \g_j $$
or in matrix form
$$ A \cdot P = \mathbf 1 \, \l ,$$
where
\begin{equation} \label{e:amatrix}
   A:=
   \begin{pmatrix}
   1 & -\g_{2} & \cdots & -\g_{N }\\
  -\g_{1} & 1 &   \cdots & -\g_{N } \\
   \vdots &   & \ddots & \vdots \\
    -\g_{1} & -\g_{2} &   \cdots & 1\\
 \end{pmatrix},
 \end{equation}
and $\mathbf 1=[1, \ldots, 1]^T$. Note that the matrix
$A$ is clearly full rank, and hence invertible.

As in Example~4, we define now an iterative pricing mechanism $\mc M^p$ such that
\begin{eqnarray} \label{e:iterative3a}
 \l (n+1) = \l (n) + \kappa_D \big( \sum_i x_i - C \big) , \\ 
  P(n+1) = (A)^{-1} \mathbf 1 \, \l(n),
\end{eqnarray}
and
\begin{equation} \label{e:iterative3b}
 x_i (n+1) =  x_i(n) - \kappa_i \dfrac{\partial J_i}{\partial x_i }=U_i^{\prime}(\g_i(n)) - P_i(n) \;\; \forall i \in \mc A,
\end{equation}
where the players adopt a gradient best response for convergence
purposes. Here, $\kappa_D$ and $\kappa_i$ denote the step sizes of the designer and player $i$, respectively.
Based on the analysis above, the mechanism is preference-compatible and
efficient. Since the players have no incentive to deviate from their (gradient) best responses,
it is also inherently strategy-proof as discussed in Example 4. This result is summarized in the following theorem.

\begin{thm} \label{thm:pricing1}
The unique equilibrium outcome of the pricing mechanism $\mc M^p$ defined by (\ref{e:iterative3a})-(\ref{e:iterative3b})
is preference-compatible, strategy-proof, and efficient.
\end{thm}

The implementation of mechanism $\mc M^p$ requires minimum information overhead. The designer
only needs to observe the aggregate received power level $ \sum_i x_i $ and the individual SIRs, $\g$,
of players both of which are already available. The player $i$, in return only needs to know the current
price $P_i$ and SIR $\g_i$ to be able to compute the (gradient) best response (see Figure~\ref{fig:pricingmech1}
for visualization). Finally, the computation of actual uplink power levels $p$ can be computed from $x$
using the measured channel gains.

\subsection*{Example 5:}

The iterative pricing mechanism $\mc M^p$ 
is illustrated with a numerical example. $10$ players with the utility parameters
$$\a= [0.23 \, 1.33 \,   0.73 \, 0.28 \,  1.13 \, 1.65 \,  1.35 \,  2.00 \, 1.92 \,  0.12] ,$$
update their power levels according to (\ref{e:iterative3b}) at each time step $n \geq 1$ with a 
stepsize of $\kappa_i=0.05 \; \forall i$. The designer, on the other hand, updates the Lagrangian multipler
$\l$ and prices $P$ based on (\ref{e:iterative3a}), where $C=5$ and $\kappa_D=0.01$. The background
noise parameter in (\ref{e:sir2}) is $\sigma=0.5$. The convergence of the mechanism $\mc M^p$ 
summarized in Algorithm~\ref{alg:iterative} is depicted
in Figures~\ref{fig:xlevels1} and \ref{fig:lambda1}.

\begin{figure}[htp]
  \centering
  \includegraphics[width=0.9\columnwidth]{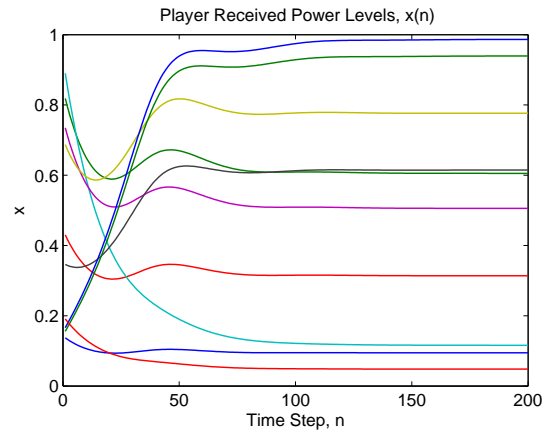}
  \caption{The evolution of user power levels, $x$, which are updated by players (\ref{e:iterative3b}) under mechanism $\mc M^p$. }
\label{fig:xlevels1}
\end{figure}
\begin{figure}[htbp]
  \centering
  \includegraphics[width=0.9\columnwidth]{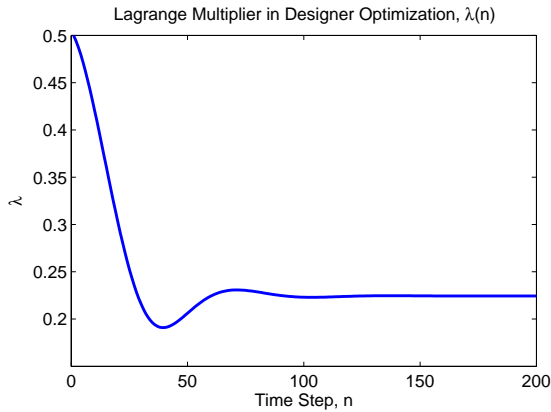}
  \caption{The evolution of Lagrange multipler, $\l$, used in computing player prices in (\ref{e:iterative3a}).}
\label{fig:lambda1}
\end{figure}

\begin{algorithm}[htp]
  \SetAlgoLined
  \KwIn{\textit{Designer (base station)}: Interference target $C$ and objective $\sum_i U_i$}
  \KwIn{\textit{Players (users)}: Utilities $U_i=\a_i \log (\g_i(x)), \forall i$}
  \KwResult{Power levels $x$ and  SIRs $\g(x)$}
  
    Initial power levels $x(0)$ and  prices $P_i(0)$ \;
    \Repeat{end of iteration}{
   \Begin(\textit{Designer:}){
    Observe player power levels $x$ \;
    Compute the matrix matrix $A$ in (\ref{e:amatrix}) \;
    Update $\l$ and prices $P$ according to (\ref{e:iterative3b}) \;
    }
   \Begin(\textit{Players:}){
    \ForEach{player $i$}{
      Estimate marginal utility $\partial U_i(x)/ \partial x_i$ \;
      Compute power level $x_i$ from (\ref{e:iterative3a}) \;
    }  
   }
    }
  \caption{Iterative Pricing Mechanism  $\mc M^p$ } \label{alg:iterative}
\end{algorithm}

\newpage
\section{Discussion and Literature Review} \label{sec:background}

There is a rich literature on Mechanism design both in the field of economics \cite{maskin1} and  recently in engineering~\cite{johari1,hajek1,lazarSemret1998,rajiv1,wuWangLiuClancy2009,huangBerryHonig2006b}.
The auction-based mechanism framework presented in Section~\ref{sec:auction} is based in principle on
progressive second price (PSP) auctions~\cite{lazarSemret1998,huangBerryHonig2006b,caines1}. 
The framework, one the one hand, simplifies PSP auctions by considering the users demanding as much of the resources 
as possible, which is a reasonable assumption in many cases since players often cannot estimate their demand accurately. On the other hand, it presents a unifying optimization framework which also
allows analysis and design of games with non-separable player utilities.

The literature on pricing  schemes is even richer than mechanism design one, especially in the networking 
community (see e.g. \cite{tansuphd,srikantbook} and references therein). The pricing mechanism framework in
Section~\ref{sec:pricing} extends those results by building on~\cite{gamenetsne,cdc09lacra},
and taking into account all of the criteria in Table~\ref{tbl:mechdesign}. 
Among other things, the presented framework captures different types of global objectives, 
e.g. quality-of-service regions, information limitations, and system dynamics. The fact
that an iterative pricing scheme similar to the one in~\cite{kelly1} 
is required to satisfy all three criteria in Table~\ref{tbl:mechdesign} is an interesting result.
This can be attributed to the designer having less leverage (no explicit resource allocation)
in pricing mechanisms compared to auction-based ones.

There are many impossibility results in the mechanism design literature \cite{holger-allerton,hurwicz1972,dasguptaHammondMaskin1979,zhou1991}. The framework presented in this paper
does not actually contradict these results for in many cases analyzed one of the criteria 
in Table~\ref{tbl:mechdesign} is achieved only approximately. Similar approximations are quite
common in game theory literature, e.g. $\eps$-NE. Hence, such relaxations are part of 
the constructive approach adopted here, and show its value.

We present next a brief survey of the literature on auctions, pricing, and mechanism design in general. 

\subsection*{Literature Review}

\subsubsection*{Auctions and Pricing in Games}
The book \cite{algorithmicGameTheory2007} provides a good overview of a variety of topics ranging from mechanism design, inefficiency of the equilibria, preference-compatibility issues and certain types of auctions. Lazar and Semret \cite{lazarSemret1998} have shown that a certain form of the Nash equilibrium holds when the progressive second price auction is applied by independent sellers on each link of a network with arbitrary topology.

Wu et al. \cite{wuWangLiuClancy2009} have proposed a repeated spectrum sharing game with cheat-proof strategies. By using the punishment-based repeated game, users get the incentive to share the spectrum in a cooperative way; and through mechanism-design-based and statistics-based approaches, user honesty is further enforced. Sengupta and Chaterjee \cite{senguptaChaterjee2009} have presented an economic framework that can be used to guide the dynamic spectrum allocation process and the service pricing mechanisms that the providers can use. They have demonstrated how pricing can be used as an effective tool for providing incentives to the providers to upgrade their network resources and offer better services. Keon and Anandalingam \cite{keonAnandalingam2003} have formulated the optimal pricing problem as a nonlinear integer expected revenue optimization problem. They simultaneously solve for prices and the resource allocations necessary to provide connections with guaranteed QoS. Maille and Tuffin \cite{mailleTuffin2006} have analyzed a multi-bid auction scheme where users compete for bandwidth at a link by submitting e.g. amount of bandwidth asked, associated unit price so that the link allocates the bandwidth and computes the charge according to the second price principle. In this case, the backbone network is overprovisioned and the access networks have a tree structure. The works \cite{yangPrasadWang2009, niyatoHossain2008, curescuTehrani2008} have discussed other interesting approaches in relation to auctions and bidding algorithms.

\subsubsection*{Strategy Proofness and Efficiency}

The property of \emph{strategy-proofness} is a fairly restrictive property. When it is combined with the property of \emph{efficiency}, this often leads to special solutions. Hurwicz \cite{hurwicz1972} has shown that there is no \emph{strategy-proof}, \emph{efficient} and \emph{individually rational} mechanism in $2$ user $2$ resource pure exchange economy. Dasgupta et al. \cite{dasguptaHammondMaskin1979} have attempted to replace \emph{individual rationality} in Hurwicz's result with a weaker axiom of \emph{non--dictatorship}. Ameliorating upon both results,
Zhou \cite{zhou1991} has established an impossibility result that there is no \emph{strategy-proof}, \emph{efficient} and \emph{non--dictatorial} mechanism in $2$ user $m$ resource ($m\geq 2$) pure exchange economies. He conjectures that there are no \emph{strategy-proof}, \emph{efficient} and \emph{non--inversely dictatorial} mechanisms in the case of $3$ or more users. In \cite{katoOhseto2002}, Zhou's conjecture has been examined and a new class of \emph{strategy-proof} and \emph{efficient} mechanisms in the case of four or more users (operators) are discovered. 

\subsubsection*{Mechanism Design in Wireless Networks}

Huiping and Junde \cite{huipingJunde2004} have proposed a \emph{strategy-proof} trust management system in the
context of wireless ad-hoc networks. This system is preference-compatible in which nodes can honestly report trust evidence and truthfully compute and broadcast trust value of themselves and other nodes. Pal and Tardos \cite{palTardos2003} have developed a general method for turning a primal-dual algorithm into a group \emph{strategy-proof} cost-sharing mechanism. The method was used to design approximately budget-balanced cost sharing mechanisms for two NP-complete problems: metric facility location, and single source rent-or-buy network design. Both mechanisms are competitive, group \emph{strategy-proof} and recover a constant fraction of the cost. The works \cite{suriInfocom2006, suriNarahariManjunath2006} 
have presented a game theoretic framework for truthful broadcast protocol and \emph{strategy-proof} pricing mechanism. Guanxiang et al. \cite{guanxiangYanZongkaiWenqing2004} have proposed an auction-based admission control and pricing mechanism for priority services, where higher priority services are allocated to the users who are more sensitive to delay, and each user pays a congestion fee for the external effect caused by their participation. The mechanism is proved to be \emph{strategy-proof} and \emph{efficient}. Wang and Li \cite{wangLi2004} have addressed the issue of user cooperation in selfish and rational wireless networks using an incentive approach. They have presented a \emph{strategy-proof} pricing mechanism for the unicast problem and given a time optimal method to compute the payment in a centralized manner and discussed implementation of the algorithm in a distributed manner. In addition, they have presented a truthful mechanism when a node only colludes with its neighbors. Garg et al. \cite{gargNarahariGujar1, gargNarahariGujar2} have provided a tutorial on mechanism design and attempted to apply it to various concepts in engineering. Huang et al. \cite{huangBerryHonig2006a, huangBerryHonig2006b} have utilized SIR and power auctions to allocate resources in a wireless scenario and presented an asynchronous distributed algorithm for updating power levels and prices to characterize convergence using supermodular game theory. Wu et al. \cite{wuWangLiuClancy2009} have proposed a repeated spectrum sharing game with cheat-proof strategies. They have proposed specific cooperation rules based on maximum total throughout and proportional fairness criteria. Sharma and Teneketzis \cite{sharmaTeneketzis2009} have presented a decentralized algorithm to allocated transmission powers, such that the algorithm takes into account the externality generated to the other users, satisfies the informational constraints of the system, and overcomes the inefficiency of pricing mechanisms.

\subsubsection*{Interference Coupling}

An axiomatic approach to interference functions has been proposed by Yates in \cite{yates1995} with extensions in \cite{huangYates1998, leungSungWongLok2004}. The Yates framework of \emph{standard interference functions} is general enough to incorporate cross-layer effects and it serves as a theoretical basis for a variety of algorithms. Certain examples include: beamforming \cite{bengtsson01}, CDMA \cite{ulukus98}, base station assignment, robust design and networking \cite{bocheSchubert_ITNet2008}. The framework can be used to combine power control and adaptive receiver strategies. Certain examples, where this has been successfully achieved are as follows. In \cite{bambosChenPottie2000} it has been proposed to incorporate admission control to avoid unfavorable interference scenarios. In \cite{xiaoShroffChong2003} the QoS requirements have been adapted to certain network conditions. In \cite{koskieGajic2005} a power control algorithm using fixed-point iterations has been proposed for a modified cost function, which permits control of convergence behavior by adjusting fixed weighting parameters.


\section{Conclusions} \label{sec:conclusion}

An unified framework is presented for developing mechanisms such
as auctions and pricing schemes, which is applicable to
a fairly general class of strategic (noncooperative) games on
networked systems. It has been shown that although the participating players of these mechanisms
are selfish, the outcome  is optimal with respect to a global criterion (e.g.
maximizing a social welfare function), preference-compatible, and
strategy-proof. The mechanism designer achieves these objectives by
imposing rules and prices to the players. In auction-based
 mechanisms the designer explicitly allocates the resources based on bids of
the participants in addition to setting prices. In pricing
mechanism, however, global objectives are enforced by only
charging the players for the resources they used. The unified
framework as well as its information structures are illustrated
through specific example resource allocation problems from wireless
and wired networks.

The presented mechanism design framework can be extended in multiple directions.
One immediate extension is multiple decision variables. A related but more challenging
extension is multi-criteria decision making, where preferences are not simply expressed
through scalar-valued utility or objective functions. Some of the other open research directions 
follow directly from relaxing the assumptions in Section~\ref{sec:model}. Improving the robustness of the 
incentive mechanisms against malicious units who do not follow the rules
and detection of such misbehavior is of both practical and theoretical interest.
In parallel, the relaxation of the assumption on designer's honesty leads
to similarly interesting questions such as how can a unit detect and respond to
misbehavior (e.g. unfairness) of the designer. Additional future research directions
include a more precise quantification of asymptotic approximations in the paper
and analysis of networking effects between players.

%

\section*{Acknowledgements}

This work is supported in part by Deutsche Telekom Laboratories. The authors wish to thank
Lacra Pavel as a collaborator of the ongoing research project.

\bibliographystyle{IEEEtran}


\end{document}